\documentclass[11pt]{article}

\usepackage{fullpage}
\usepackage{amsmath}
\usepackage{amsfonts}
\usepackage{amssymb}
\usepackage{amsthm}
\usepackage{mathrsfs}
\usepackage{xspace}
\usepackage{booktabs}
\usepackage{array}
\usepackage{longtable}
\usepackage[dvipsnames]{xcolor}
\usepackage{hyperref}
\usepackage[nottoc]{tocbibind}

\numberwithin{equation}{section}
\allowdisplaybreaks

\newtheorem{theorem}{Theorem}
\newtheorem{lemma}{Lemma}
\newtheorem{proposition}{Proposition}
\newtheorem{corollary}{Corollary}

\theoremstyle{definition}
\newtheorem{definition}{Definition}

\theoremstyle{remark}

\newcommand{\SC}{\operatorname{SC}}
\newcommand{\OPT}{\operatorname{OPT}}
\newcommand{\CMP}{\mathsf{CMP}}
\newcommand{\CM}{\mathsf{CM}}
\newcommand{\R}{\mathbb{R}}
\newcommand{\one}{\mathbf{1}}

\title{Learning-Augmented Strategyproof Facility Location in $\mathbb{R}^d$ with $\ell_p$ Distances}

\author{Hau Chan$^{1}$ \qquad Jianan Lin$^{2}$ \qquad Chenhao Wang$^{3,4}$\\[0.6em]
\small $^{1}$University of Nebraska--Lincoln\\
\small $^{2}$Rensselaer Polytechnic Institute\\
\small $^{3}$Beijing Normal University--Zhuhai\\
\small $^{4}$Beijing Normal--Hong Kong Baptist University}

\date{}

\begin{document}
\maketitle

\begin{abstract}
We study learning-augmented mechanism design for locating a single facility in $\mathbb{R}^d$ to minimize the sum of the agents' $\ell_p$ distances to the facility.
We analyze the coordinate-wise median with predictions (CMP) mechanism, which adds $cn$ virtual agents at a predicted optimal facility location to the $n$ reported locations and returns their coordinate-wise median.
The parameter $c\in[0,1)$ represents confidence in the prediction; CMP is known to be strategyproof when both are fixed independently of the reports.
We determine its approximation guarantees under correct predictions (consistency) and arbitrary predictions (robustness) in two settings.
First, for $d=2$, we establish exact guarantees for every $p\in[1,+\infty]$.
For $1<p<+\infty$, the consistency and robustness are respectively
$\Bigl(1+\bigl(\frac{1-c}{1+c}\bigr)^{\frac{p}{p-1}}\Bigr)^{\frac{p-1}{p}}$ and
$\Bigl(1+\bigl(\frac{1+c}{1-c}\bigr)^{\frac{p}{p-1}}\Bigr)^{\frac{p-1}{p}}$.
We use the median condition in each coordinate to compare the mechanism's cost with the optimal cost, and establish tightness using instances with agents at only three distinct locations.

Second, for $1<p<+\infty$, we obtain dimension-independent consistency and robustness upper bounds valid for every $d\ge1$.
For each fixed $p$ and $c$, we construct families of instances whose ratios approach the respective upper bounds as $d\to\infty$, proving asymptotic tightness.
We also establish exact guarantees for $p=1$ in every dimension and asymptotically tight bounds for $p=+\infty$.
Our high-dimensional results recover the prediction-free bounds of Gravin and Jia (STOC 2025) when $c=0$ and their learning-augmented bounds in arbitrary-dimensional Euclidean spaces when $p=2$.
\end{abstract}

\section{Introduction}\label{sec:introduction}

Approximate mechanism design without money studies optimization problems in which part of the input is privately held by strategic agents and monetary transfers are unavailable \cite{procaccia2013approximate}.
Agents may misreport this information to obtain outcomes that better serve their own interests.
The goal is therefore to design strategyproof mechanisms, under which no agent benefits from misreporting, while providing good approximation guarantees for the optimization objective.
Problems studied in this framework include facility location \cite{procaccia2013approximate}, scheduling \cite{koutsoupias2014scheduling}, and one-sided matching \cite{bhalgat2011social}.

\paragraph{Strategyproof Facility Location.}
Facility location is a classical optimization problem in operations research, economics, and computer science, and a motivating case study for approximate mechanism design without money \cite{drezner2004facility,procaccia2013approximate,chan2021mechanismsurvey}.
A social planner chooses a location for a public facility, such as a school, park, or library, based on the locations reported by a set of agents.
Each agent incurs the distance from her true location to the facility and may benefit from misreporting if doing so moves the facility closer to her.
Under the utilitarian social-cost objective, the planner seeks to minimize the sum of these distances while ensuring strategyproofness.
Beyond physical facilities, the model also captures preference aggregation, such as choosing a collective policy position from committee members' preferred positions \cite{black1948rationale,moulin1980strategy}.

\paragraph{Multi-Dimensional Facility Location.}
The classical one-dimensional problem provides an important benchmark.
Locating the facility at a median report is strategyproof and minimizes the utilitarian social cost \cite{moulin1980strategy,procaccia2013approximate}.
In $\mathbb{R}^d$ with $\ell_p$ distances, where $d\ge1$ and $p\in[1,+\infty]$, the coordinate-wise median (CM) mechanism takes a median of the reports in each coordinate.
This mechanism remains strategyproof and is optimal for $p=1$, but need not minimize the utilitarian social cost when $d\ge2$ and $p>1$.
For Euclidean distance ($p=2$), Meir \cite{DBLP:conf/sagt/Meir19} proved a $\sqrt d$ approximation bound.
Gravin and Jia \cite{gravin2025approximation} improved this to the dimension-independent bound $\sqrt{6\sqrt3-8}\approx1.55$ and established its asymptotic tightness as $d\to\infty$.
They also obtained dimension-independent, asymptotically tight bounds for general $\ell_p$ distances.
For $d=2$, Hastings \cite{hastings2026strategic} determined the exact CM approximation ratio under the utilitarian objective for every $\ell_p$ distance.
Beyond the utilitarian objective, Chan, Lin, and Wang \cite{chan2026strategyproof} and Hastings \cite{hastings2026strategic} also analyzed CM when agents' costs are aggregated by a norm, with Euclidean distances in the former work and general $\ell_p$ distances in the latter.
These results quantify the worst-case loss of using CM without additional information about an optimal facility location.

\paragraph{Learning-Augmented Facility Location.}
The learning-augmented approach uses predictions, derived for example from historical data or domain knowledge, to improve performance while retaining worst-case guarantees when the predictions are inaccurate \cite{lykouris2021competitive,mitzenmacher2022algorithms}.
Agrawal et al.\ \cite{agrawal2022learning} introduced this approach to strategic facility location, where a mechanism receives a prediction of an optimal facility location.
Its approximation guarantee when the prediction is correct is called \emph{consistency}, and its guarantee under arbitrary predictions is called \emph{robustness}.
The challenge is to benefit from accurate predictions while preserving both strategyproofness and bounded robustness.

For the utilitarian objective in the Euclidean plane ($d=2$ and $p=2$), Agrawal et al.\ \cite{agrawal2022learning} proposed the coordinate-wise median with predictions (CMP) mechanism and determined its exact consistency and robustness.
Given $n$ reported locations and a confidence parameter $c\in[0,1)$, CMP adds $cn$ virtual agents at the predicted location and returns the coordinate-wise median.
The mechanism is strategyproof when the prediction and $c$ are fixed independently of the reports.
When $c=0$, it reduces to CM, whose consistency and robustness are both $\sqrt2$ in this setting; increasing $c$ improves consistency at the expense of robustness.
Gravin and Jia \cite{gravin2025approximation} subsequently obtained consistency and robustness guarantees for the same mechanism in arbitrary-dimensional Euclidean spaces ($p=2$).
These results motivate extending the analysis beyond Euclidean distance: although prediction-free CM has been analyzed for general $\ell_p$ distances, the corresponding consistency and robustness of CMP remain to be determined.
We therefore ask:

\begin{quote}
\emph{What consistency and robustness does the coordinate-wise median with predictions mechanism achieve in $\mathbb{R}^d$ with $\ell_p$ distances under the utilitarian social-cost objective?}
\end{quote}

\subsection{Our Contributions}\label{subsec:contributions}

We denote the CMP mechanism with confidence parameter $c\in[0,1)$ by $\CMP_c$.
Our two main contributions are as follows.

\paragraph{Tight Bounds in Two Dimensions.}
In $\mathbb{R}^2$ with $\ell_p$ distances, we determine the exact consistency and robustness of $\CMP_c$, denoted by $A_p(c)$ and $B_p(c)$, respectively.
For $1<p<+\infty$, they are
\begin{align*}
 A_p(c)=\left(1+\left(\frac{1-c}{1+c}\right)^{\frac{p}{p-1}}\right)^{\frac{p-1}{p}},
 \qquad
 B_p(c)=\left(1+\left(\frac{1+c}{1-c}\right)^{\frac{p}{p-1}}\right)^{\frac{p-1}{p}}.
\end{align*}
We also prove the endpoint values $A_1(c)=1$, $B_1(c)=\frac{1+c}{1-c}$, $A_\infty(c)=\frac{2}{1+c}$, and $B_\infty(c)=\frac{2}{1-c}$.
For $p=2$, these expressions coincide with the planar Euclidean guarantees of Agrawal et al.\ \cite{agrawal2022learning}.
When $c=0$, the mechanism ignores the prediction and reduces to CM; both guarantees then recover its exact planar approximation ratio $2^{1-\frac{1}{p}}$ for finite $p$, and $2$ for $p=+\infty$ \cite{hastings2026strategic}.

For the upper bounds, we divide the agents into four groups according to which side of the mechanism's output they lie on in each coordinate.
Within each group, H\"older's inequality relates the agents' costs at the output to their costs at an optimal facility.
The median condition in each coordinate limits the number of agents on either side; combining these restrictions with the cost inequalities gives the consistency and robustness upper bounds.
To prove tightness, we construct finite instances with agents at at most three distinct locations whose cost ratios attain or approach the respective bounds arbitrarily closely.
Thus neither guarantee can be improved for this mechanism.

\paragraph{Asymptotically Tight Bounds in High Dimensions.}
For $1<p<+\infty$, we prove consistency and robustness upper bounds that hold in $\mathbb{R}^d$ for every $d\ge1$ and do not depend on $d$.
To state these bounds, for $\mu\in(0,1)$ define
\begin{align*}
 \Gamma_p(\mu)=\left(1+\left[\max_{0\le a\le\mu}\frac{a^{1/p}+(\mu-a)/(1-\mu)}{(1-a)^{1/p}}\right]^{\frac{p}{p-1}}\right)^{\frac{p-1}{p}}.
\end{align*}
Substituting $\mu=\frac{1-c}{2}$ and $\mu=\frac{1+c}{2}$ gives the consistency upper bound $\Gamma_p\bigl(\frac{1-c}{2}\bigr)$ and the robustness upper bound $\Gamma_p\bigl(\frac{1+c}{2}\bigr)$, respectively.
For each fixed $p$ and $c$, we construct families of finite instances whose corresponding cost ratios converge to these bounds as $d\to\infty$, proving asymptotic tightness.
For $p=1$, the exact consistency and robustness are $1$ and $\frac{1+c}{1-c}$ in every dimension.
For $p=+\infty$, we obtain the asymptotically tight consistency and robustness upper bounds $\frac{3+c}{1+c}$ and $\frac{3-c}{1-c}$, respectively.
When $c=0$, our results recover the prediction-free, dimension-independent bounds of Gravin and Jia \cite{gravin2025approximation} for general $\ell_p$ distances; when $p=2$, they recover their learning-augmented guarantees in arbitrary-dimensional Euclidean spaces.

For the upper bounds, we associate a number with each agent that records which coordinates place her on the same side of the output as an optimal facility, weighted by that facility's displacement in each coordinate.
The coordinate-wise median conditions bound the average of these numbers, allowing us to compare social costs through a function of a single variable.
A convexity argument reduces the worst-case comparison to distributions on at most two values and yields the maximization in $\Gamma_p$.
For the matching lower bounds, we construct instances that are symmetric under permutations of the coordinates and whose cost ratios approach the upper bounds as the dimension grows.

\subsection{Additional Related Work}\label{subsec:related}

Moulin \cite{moulin1980strategy} characterized generalized median voter schemes on one-dimensional single-peaked domains.
Border and Jordan \cite{border1983} studied multidimensional characterizations under separable preferences, relating strategyproof mechanisms satisfying unanimity to coordinate-wise rules with phantom voters.
Procaccia and Tennenholtz \cite{procaccia2013approximate} developed approximate mechanism design without money, using facility location to study the efficiency loss imposed by strategyproofness.
For a broader account of facility-location mechanism design and its extensions, we refer readers to the survey of Chan et al.\ \cite{chan2021mechanismsurvey}.

\paragraph{Strategyproof Facility Location without Predictions.}
The literature extends the classical single-facility problem along several directions.
For multiple facilities, where agents are served by their nearest facility, Fotakis and Tzamos studied concave connection costs and winner-imposing mechanisms \cite{fotakis2013strategyproof,fotakis2013winner}.
Recent work develops improved randomized mechanisms for two and more facilities \cite{ma2026breaking,aziz2026anchoring,chan2026randomized,jia2026product}.
Other studies change the location space, examining strategyproof mechanisms on networks \cite{schummer2002strategy,alon2010strategyproof}, discrete lines and cycles \cite{dokow2012mechanism}, and approximation guarantees on cycles \cite{rogowski2025improved}.
Preference variants include obnoxious facilities, which agents prefer to be farther away from, and hybrid models combining attractive and obnoxious preferences \cite{cheng2013strategy,Ibara2012characterize,feigenbaum2015strategyproof,chan2025obnoxious}.
Capacity-constrained models additionally limit how many agents each facility can serve \cite{aziz2020facility,aziz2020capacity}.
Alternative objectives include least-squares cost \cite{feldman2013strategyproof}, norm-based social costs on the line \cite{feigenbaum2017approximately}, and equitable outcomes \cite{walsh2025equitable}.
These variants change the service rule, location space, preferences, or objective; our setting retains a single attractive facility and the utilitarian objective in $\mathbb{R}^d$.

\paragraph{Learning-Augmented Facility Location.}
Agrawal et al.\ \cite{agrawal2022learning} and Xu and Lu \cite{Xu2022} introduced predictions into facility-location mechanism design.
For maximum cost, Balkanski et al.\ \cite{balkanski2024randomized} studied randomized mechanisms under different prediction models, while Chan, Lin, and Wang \cite{chan2025prediction} considered the plane with general $\ell_p$ distances.
Other extensions include candidate facility locations \cite{DBLP:conf/tamc/FangFLN24}, obnoxious facilities \cite{DBLP:journals/corr/abs-2212-09521}, and the envy-ratio objective \cite{nips2025envy}.
The information supplied to a mechanism also varies across models.
Barak, Gupta, and Talgam-Cohen \cite{DBLP:conf/nips/Barak0T24} studied mostly and approximately correct (MAC) predictions of agents' locations, allowing some predictions to be arbitrarily wrong and bounded errors in the rest.
Christodoulou, Sgouritsa, and Vlachos \cite{DBLP:conf/nips/000SV24} studied mechanisms augmented with output recommendations, measuring performance through the quality of the recommended outcome.
For utilitarian facility location, the planar Euclidean analysis of Agrawal et al.\ \cite{agrawal2022learning} and the high-dimensional Euclidean analysis of Gravin and Jia \cite{gravin2025approximation} are the closest to our work.
We retain their prediction of an optimal facility location and extend the analysis of CMP to general $\ell_p$ distances.

\section{Preliminaries}\label{sec:preliminaries}

\subsection{Settings}

We study the problem of locating a single facility in $\ell_p(\R^d)$, where $d\ge1$ and $p\in[1,+\infty]$.
For every finite population size $n\ge1$, let $N=\{1,\ldots,n\}$ be the set of agents.
Agent $i\in N$ has a preferred location $x_i=(x_{i1},\ldots,x_{id})\in\R^d$, and $\mathbf{x}=(x_1,\ldots,x_n)\in(\R^d)^n$ denotes the location profile.
The distance between two points $x,y\in\R^d$ is measured by the $\ell_p$ norm, given by
\begin{align*}
 \|x-y\|_p=\biggl(\sum_{j=1}^d|x_j-y_j|^p\biggr)^{\frac{1}{p}}
 \qquad\text{for }1\le p<+\infty.
\end{align*}
In particular, for $p=+\infty$, we use the limiting norm $\|x-y\|_\infty=\max_{1\le j\le d}|x_j-y_j|$.
The choices $p=1$, $p=2$, and $p=+\infty$ correspond to the sum of the coordinate distances, Euclidean distance, and the maximum coordinate distance, respectively.

Given a facility location $y\in\R^d$, the cost incurred by agent $i$ is her distance to the facility, namely $\|x_i-y\|_p$.
Our objective is to minimize the \emph{utilitarian social cost}, defined by
\begin{align*}
 \SC_p(\mathbf{x},y)=\sum_{i\in N}\|x_i-y\|_p.
\end{align*}
The optimal social cost on profile $\mathbf{x}$ is
\begin{align*}
 \OPT_p(\mathbf{x})=\min_{y\in\R^d}\SC_p(\mathbf{x},y).
\end{align*}
We write $f\in\arg\min_{y\in\R^d}\SC_p(\mathbf{x},y)$ for an optimal facility location, so that $\SC_p(\mathbf{x},f)=\OPT_p(\mathbf{x})$.
An optimal facility need not be unique.

In the learning-augmented setting, the mechanism additionally receives a prediction $\widehat f\in\R^d$ of an optimal facility location. 
The prediction may be inaccurate, and we impose no probabilistic assumption on its error.
A deterministic mechanism specifies, for each population size $n$, a function $M_n:(\R^d)^n\times\R^d\to\R^d$ that maps the reported profile and the prediction to a facility location.
We suppress the subscript $n$ when the population size is clear and write $M(\mathbf{x},\widehat f)$ for the output.
The prediction is exogenous: it is held fixed when an agent changes her report.

Agents may strategically report locations different from their true locations in order to obtain a facility closer to themselves.
For an agent $i\in N$ and an alternative report $x_i'\in\R^d$, let $(x_i',\mathbf{x}_{-i})$ denote the profile obtained from $\mathbf{x}$ by replacing only agent $i$'s report with $x_i'$.
A mechanism $M$ is \emph{strategyproof} if, for every profile $\mathbf{x}$, prediction $\widehat f$, agent $i\in N$, and alternative report $x_i'$, we have
\begin{align*}
 \bigl\|x_i-M(\mathbf{x},\widehat f)\bigr\|_p
 \le \bigl\|x_i-M((x_i',\mathbf{x}_{-i}),\widehat f)\bigr\|_p.
\end{align*}
Thus, truthful reporting is a weakly dominant strategy, regardless of the accuracy of the prediction.
A mechanism is \emph{anonymous} if permuting the agents' reports leaves its output unchanged, with the prediction fixed.
It is \emph{unanimous} if it returns $a$ whenever all agents report the same point $a\in\R^d$, for every prediction.

We evaluate a mechanism by comparing its social cost $\SC_p(\mathbf{x},M(\mathbf{x},\widehat f))$ with $\OPT_p(\mathbf{x})$.
For profiles with positive optimal cost, their ratio is the approximation ratio on that instance.
The two performance guarantees of interest are \emph{consistency} and \emph{robustness}. 
A mechanism is $\alpha$-consistent if
\begin{align*}
 \SC_p(\mathbf{x},M(\mathbf{x},\widehat f))\le\alpha\OPT_p(\mathbf{x})
\end{align*}
whenever $\widehat f$ is an optimal facility.
This requirement applies to every optimal prediction when the optimal facility is not unique.
A mechanism is $\beta$-robust if
\begin{align*}
 \SC_p(\mathbf{x},M(\mathbf{x},\widehat f))\le\beta\OPT_p(\mathbf{x})
\end{align*}
for every prediction $\widehat f\in\R^d$.
Both guarantees must hold uniformly over all finite population sizes and all location profiles in the specified space.
Consistency measures performance with a correct prediction, whereas robustness bounds the loss even when the prediction is arbitrarily inaccurate.

Any mechanism with finite robustness must be unanimous.
Indeed, when all agents are located at $a$, the optimal social cost is zero, and the robustness inequality forces the mechanism to return $a$.
Conversely, $\OPT_p(\mathbf{x})=0$ only when all agents share a location.
The mechanism studied below is unanimous, so we restrict the suprema defining its approximation ratios to profiles with positive optimal cost.

\subsection{CMP Mechanism}

The \emph{coordinate-wise median} mechanism $\CM$ returns the median of the reports in each coordinate, choosing the lower of the two middle values when their number is even.
The coordinate-wise median with predictions incorporates a prediction by adding virtual agents at the predicted location \cite{agrawal2022learning}.

\begin{definition}[Coordinate-Wise Median with Predictions]\label{def:cmp}
Fix a confidence parameter $c\in[0,1)$ independently of the reports.
Given a profile $\mathbf{x}$ of $n$ agents and a prediction $\widehat f$, the mechanism $\CMP_c$ adds $cn$ virtual agents at $\widehat f$ and returns the coordinate-wise median of the augmented profile, using the lower median in each coordinate.
\end{definition}

For arbitrary $cn$, the virtual agents represent total weight $cn$ at $\widehat f$, while each real agent has weight one; the lower median is the first coordinate value whose cumulative weight reaches at least $\frac{(1+c)n}{2}$.
Virtual agents are not included in the social cost.
When $c=0$, we recover $\CMP_0=\CM$; increasing $c$ gives the prediction more weight.
Since $c<1$, real agents at a common location form a strict majority of the augmented profile, so $\CMP_c$ is unanimous.

For fixed $p$, $d$, and $c$, we denote the exact consistency and robustness of $\CMP_c$ by
\begin{align*}
 \alpha_{p,d}(c)&=\sup_{\mathbf{x},\,\widehat f\in\arg\min_y\SC_p(\mathbf{x},y)}
 \frac{\SC_p(\mathbf{x},\CMP_c(\mathbf{x},\widehat f))}{\OPT_p(\mathbf{x})},\\
 \beta_{p,d}(c)&=\sup_{\mathbf{x},\,\widehat f\in\R^d}
 \frac{\SC_p(\mathbf{x},\CMP_c(\mathbf{x},\widehat f))}{\OPT_p(\mathbf{x})}.
\end{align*}
Both suprema range over all finite population sizes and profiles with $\OPT_p(\mathbf{x})>0$.

\begin{proposition}\label{prop:sp}
The mechanism $\CMP_c$ is deterministic, anonymous, and strategyproof.
\end{proposition}

\begin{proof}
For fixed $\widehat f$ and $c$, each coordinate is a generalized median rule with fixed virtual agents and is strategyproof on the line \cite{moulin1980strategy}.
The same argument applies when the virtual weight $cn$ is nonintegral.
To see this directly, fix a coordinate, write $a$ for the agent's true coordinate and $m$ for the truthful lower median, and hold the other reports fixed.
If $a<m$, the cumulative weight at $a$ is below the median threshold and cannot increase after a misreport, since the agent already contributes her full weight there.
Likewise, at any $y\in[a,m)$ the cumulative weight cannot increase, so the new median cannot be smaller than $m$.
If $a>m$, the other agents and the prediction already supply the threshold weight at $m$, so the new median cannot exceed $m$.
If $a=m$, her coordinate cost is already zero.
Thus, a misreport cannot reduce an agent's absolute distance to the output in any coordinate, and hence cannot reduce her $\ell_p$ cost.
Anonymity follows from invariance under permutations of the reports, and determinism follows from the lower-median convention.
\end{proof}

For the H\"older inequalities used below, we write $q=\frac{p}{p-1}$ for the conjugate exponent of $1<p<+\infty$, so that $\frac{1}{p}+\frac{1}{q}=1$.
At $p=1$ and $p=\infty$, we use $q=\infty$ and $q=1$, respectively.

\section{Tight Bounds in Two Dimensions}\label{sec:two-dimensional}

This section determines the exact consistency and robustness of $\CMP_c$ in the plane.

\begin{theorem}\label{thm:planar}
For every $c\in[0,1)$ and $p\in[1,+\infty]$, the coordinate-wise median with predictions in $\ell_p(\R^2)$ has exact consistency $A_p(c)$ and exact robustness $B_p(c)$, where, for $1<p<+\infty$,
\begin{align*}
 A_p(c)=\left(1+\left(\frac{1-c}{1+c}\right)^{\frac{p}{p-1}}\right)^{\frac{p-1}{p}},
 \qquad
 B_p(c)=\left(1+\left(\frac{1+c}{1-c}\right)^{\frac{p}{p-1}}\right)^{\frac{p-1}{p}}.
\end{align*}
At the endpoints,
\begin{align*}
 A_1(c)=1,\qquad B_1(c)=\frac{1+c}{1-c},\qquad
 A_\infty(c)=\frac{2}{1+c},\qquad B_\infty(c)=\frac{2}{1-c}.
\end{align*}
\end{theorem}

We prove the upper bounds by grouping agents according to the signs of their coordinates relative to the mechanism output.
The median condition bounds the weight on each side, and a distance inequality for each of the four groups then bounds the total cost.
We next give matching three-location instances.
Together, these results also establish the three-location reduction used in the Euclidean analysis of Agrawal et al.\ \cite[Section~4.4]{agrawal2022learning}.

For convenience, we allow positive agent weights $w_i$ with $\sum_i w_i=1$ and write social cost as $\SC_p(\mathbf{x},y)=\sum_i w_i\|x_i-y\|_p$.
The prediction has weight $c$.
Ordinary profiles correspond to $w_i=\frac{1}{n}$.
In this weighted notation, $\OPT_p(\mathbf{x})$ is the minimum of the weighted social cost; for an ordinary profile, both costs are divided by $n$, leaving their ratio and the optimal locations unchanged.
We prove the upper bounds for these weighted profiles and explicitly approximate the lower-bound constructions by ordinary finite profiles.

\subsection{Upper Bounds}

The first lemma bounds the contribution of a single agent.
We place the mechanism output at the origin and compare its cost with the cost at a point $f$ in the nonnegative quadrant.

\begin{lemma}\label{lem:planar-quadrants}
Let $1<p<+\infty$, $f=(a,b)$ with $a,b\ge0$, and $\lambda,\kappa>0$ satisfy $\lambda^{q}+\kappa^{q}=1$.
Write $F=\|f\|_p$, $L=a+b$, and $h(z)=\|z-f\|_p-\lambda\|z\|_p$.
Then
\begin{align*}
 h(z)\ge
 \begin{cases}
  \kappa F, & z_1\le0,\ z_2\le0,\\
  \kappa b-\lambda a, & z_1>0,\ z_2\le0,\\
  \kappa a-\lambda b, & z_1\le0,\ z_2>0,\\
  -\lambda F, & z_1>0,\ z_2>0.
 \end{cases}
\end{align*}
If $\kappa\le\lambda$, the first bound can be strengthened to $h(z)\ge\kappa L$.
\end{lemma}

\begin{proof}
For any $z$, the triangle inequality gives $\|z\|_p\le\|z-f\|_p+F$.
Since $\lambda<1$, it follows that
\begin{align*}
 h(z)&=\|z-f\|_p-\lambda\|z\|_p =(1-\lambda)\|z-f\|_p+\lambda(\|z-f\|_p-\|z\|_p) \ge(1-\lambda)\|z-f\|_p-\lambda F\ge-\lambda F,
\end{align*}
proving the last bound.

Suppose $z_1>0$ and $z_2\le0$, and put $T=\bigl(|z_1-a|^p+|z_2|^p\bigr)^{\frac{1}{p}}$.
Writing $z=(a,0)+(z_1-a,z_2)$, the triangle inequality gives $\|z\|_p\le a+T$.
For the distance to $f$, note that $z_2\le0$ and $b\ge0$ imply $|z_2-b|=|z_2|+b$.
We also use $(s+t)^p\ge s^p+t^p$ for $s,t\ge0$; indeed, for fixed $s$, the function $(s+t)^p-t^p$ is nondecreasing in $t$ and has value $s^p$ at zero.
Consequently,
\begin{align*}
 \|z-f\|_p^p
 &=|z_1-a|^p+|z_2-b|^p =|z_1-a|^p+(|z_2|+b)^p \ge|z_1-a|^p+|z_2|^p+b^p=T^p+b^p.
\end{align*}
Taking the $p$-th root gives $\|z-f\|_p\ge(T^p+b^p)^{\frac{1}{p}}$.
To bound this expression, recall that H\"older's inequality for nonnegative vectors $r,s\in\R^2$ states that
\begin{align*}
 r_1s_1+r_2s_2
 \le \bigl(r_1^p+r_2^p\bigr)^{\frac{1}{p}}
 \bigl(s_1^{q}+s_2^{q}\bigr)^{\frac{1}{q}},
 \qquad \frac{1}{p}+\frac{1}{q}=1.
\end{align*}
Taking $r=(T,b)$ and $s=(\lambda,\kappa)$, the second factor on the right is one by our assumption on $\lambda$ and $\kappa$.
More explicitly,
\begin{align*}
 \lambda T+\kappa b
 &\le(T^p+b^p)^{\frac{1}{p}}(\lambda^q+\kappa^q)^{\frac{1}{q}}
 =(T^p+b^p)^{\frac{1}{p}}.
\end{align*}
Combining these inequalities gives
\begin{align*}
 h(z)&=\|z-f\|_p-\lambda\|z\|_p \ge(T^p+b^p)^{\frac{1}{p}}-\lambda(T+a) \ge\lambda T+\kappa b-\lambda(T+a)
 =\kappa b-\lambda a.
\end{align*}
For the other mixed-sign case, $z_1\le0$ and $z_2>0$, exchange the roles of the coordinates and put $S=\bigl(|z_1|^p+|z_2-b|^p\bigr)^{\frac{1}{p}}$.
The decomposition $z=(0,b)+(z_1,z_2-b)$ gives $\|z\|_p\le b+S$, while $|z_1-a|=|z_1|+a$ gives $\|z-f\|_p^p\ge S^p+a^p$.
Applying H\"older's inequality to $(S,a)$ and $(\lambda,\kappa)$ therefore yields
\begin{align*}
 h(z)&\ge(S^p+a^p)^{\frac{1}{p}}-\lambda(S+b) \ge\lambda S+\kappa a-\lambda(S+b)
 =\kappa a-\lambda b.
\end{align*}

If both coordinates of $z$ are nonpositive, write $z=-v$ with $v\ge0$ coordinate-wise.
Then $\|z\|_p=\|v\|_p$ and $\|z-f\|_p=\|f+v\|_p$.
Using the same inequality for nonnegative $p$-th powers in both coordinates, we obtain
\begin{align*}
 \|f+v\|_p^p
 &=(a+v_1)^p+(b+v_2)^p \ge a^p+v_1^p+b^p+v_2^p
 =F^p+\|v\|_p^p.
\end{align*}
This time, H\"older's inequality with $r=(\|v\|_p,F)$ and $s=(\lambda,\kappa)$ gives
\begin{align*}
 \lambda\|v\|_p+\kappa F
 &\le\bigl(\|v\|_p^p+F^p\bigr)^{\frac{1}{p}}
 (\lambda^q+\kappa^q)^{\frac{1}{q}} =\bigl(\|v\|_p^p+F^p\bigr)^{\frac{1}{p}}.
\end{align*}
Thus
\begin{align*}
 h(z)&=\|f+v\|_p-\lambda\|v\|_p \ge\bigl(F^p+\|v\|_p^p\bigr)^{\frac{1}{p}}-\lambda\|v\|_p \ge\lambda\|v\|_p+\kappa F-\lambda\|v\|_p=\kappa F.
\end{align*}

It remains to prove the stronger bound when $\kappa\le\lambda$.
We continue to assume $z=-v$ with $v\ge0$.
Since $h(z)=\|f+v\|_p-\lambda\|v\|_p$, the desired bound is equivalent to $\|f+v\|_p\ge\lambda\|v\|_p+\kappa(a+b)$.
To prove it, we first establish an inequality for every nonnegative vector $y$ and every vector $u$ between $0$ and $y$ coordinate-wise, and then substitute $y=f+v$ and $u=v$.
For any nonnegative vector $y$, consider $\phi(u)=\lambda\|u\|_p+\kappa\|y-u\|_1$ on the rectangle $0\le u_j\le y_j$.
Both norm terms are convex functions of $u$, so $\phi$ is convex.
The vertices are $v^{(1)}=(0,0)$, $v^{(2)}=(y_1,0)$, $v^{(3)}=(0,y_2)$, and $v^{(4)}=y$.
Every point $u$ of the rectangle can be written as $u=\sum_{j=1}^4\theta_jv^{(j)}$ with $\theta_j\ge0$ and $\sum_{j=1}^4\theta_j=1$.
Convexity then gives
\begin{align*}
 \phi(u)
 &=\phi\left(\sum_{j=1}^4\theta_jv^{(j)}\right)
 \le\sum_{j=1}^4\theta_j\phi(v^{(j)})
 \le\max_{1\le j\le4}\phi(v^{(j)}).
\end{align*}
This argument also applies when a coordinate of $y$ is zero, in which case some vertices coincide.
It therefore suffices to show that each of the four vertex values is at most $\|y\|_p$.
At $u=0$, the value is $\kappa(y_1+y_2)$.
Since $2\kappa^{q}\le\lambda^{q}+\kappa^{q}=1$, we have $\kappa\le2^{-\frac{1}{q}}$.
Applying H\"older's inequality to $r=y$ and $s=(1,1)$ therefore gives
\begin{align*}
 \phi(0)&=\kappa(y_1+y_2) \le\kappa(y_1^p+y_2^p)^{\frac{1}{p}}(1^q+1^q)^{\frac{1}{q}} =\kappa 2^{\frac{1}{q}}\|y\|_p\le\|y\|_p.
\end{align*}
At $u=y$, the second norm term vanishes, so $\phi(y)=\lambda\|y\|_p\le\|y\|_p$.
At $u=(y_1,0)$, the two norm terms are $y_1$ and $y_2$, respectively.
H\"older's inequality with $r=y$ and $s=(\lambda,\kappa)$ gives
\begin{align*}
 \phi((y_1,0))&=\lambda y_1+\kappa y_2 \le(y_1^p+y_2^p)^{\frac{1}{p}}(\lambda^q+\kappa^q)^{\frac{1}{q}}
 =\|y\|_p.
\end{align*}
At $u=(0,y_2)$, use $s=(\kappa,\lambda)$ instead to obtain
\begin{align*}
 \phi((0,y_2))&=\kappa y_1+\lambda y_2 \le(y_1^p+y_2^p)^{\frac{1}{p}}(\kappa^q+\lambda^q)^{\frac{1}{q}}
 =\|y\|_p.
\end{align*}
All four vertex bounds hold, and hence $\lambda\|u\|_p+\kappa\|y-u\|_1\le\|y\|_p$ throughout the rectangle.
Finally, take $y=f+v$ and $u=v$; this point belongs to the rectangle because $f,v\ge0$ coordinate-wise.
Since $y-u=f$ and $\|f\|_1=a+b=L$, we obtain
\begin{align*}
 h(z)&=\|f+v\|_p-\lambda\|v\|_p \ge\bigl(\lambda\|v\|_p+\kappa\|f\|_1\bigr)-\lambda\|v\|_p
 =\kappa L,
\end{align*}
which proves the stronger bound.
\end{proof}

The next lemma combines these four bounds.
Its hypothesis concerns only the mass on each side of the origin; the comparison point need not be optimal.

\begin{lemma}\label{lem:planar-cost}
Fix $1<p<+\infty$ and $\mu\in(0,1)$.
Let $\mathbf{x}$ be a weighted profile of total weight one and let $f\in\R^2$.
Suppose that, for every coordinate $j$ with $f_j\ne0$, the total weight of agents satisfying $x_{ij}f_j>0$ is at most $\mu$.
Then
\begin{align*}
 \sum_i w_i\|x_i\|_p
 \le
 \left(1+\left(\frac{\mu}{1-\mu}\right)^{q}\right)^{\frac{1}{q}}
 \sum_i w_i\|x_i-f\|_p.
\end{align*}
\end{lemma}

\begin{proof}
Set $r=\frac{\mu}{1-\mu}$, $\rho=\bigl(1+r^{q}\bigr)^{\frac{1}{q}}$, $\lambda=\frac{1}{\rho}$, and $\kappa=r\lambda$.
These choices give the two identities
\begin{align*}
 \lambda^q+\kappa^q=\frac{1+r^q}{\rho^q}=1,
 \qquad (1-\mu)\kappa=(1-\mu)\frac{\mu}{1-\mu}\lambda=\mu\lambda.
\end{align*}
The first allows us to apply Lemma~\ref{lem:planar-quadrants}, while the second will cancel the terms involving the two side weights.
With $h$ as in that lemma, we have 
\begin{align*}
 \sum_iw_i h(x_i)=\sum_iw_i\|x_i-f\|_p-\lambda\sum_iw_i\|x_i\|_p.
\end{align*}
Thus proving that this sum is nonnegative and dividing by $\lambda>0$ gives the claimed factor $\rho$.

Reflect coordinates of both the reports and $f$ as needed so that $f=(a,b)$ is nonnegative.
This preserves the distances and the hypothesis about weights on the same side as $f$.
If $f=0$, the two costs coincide and the conclusion follows from $\rho\ge1$.
If $a=0$ and $b>0$, Lemma~\ref{lem:planar-quadrants} gives $h(x_i)\ge\kappa b$ for $x_{i2}\le0$ and $h(x_i)\ge-\lambda b$ otherwise.
Let $\eta$ be the weight with $x_{i2}>0$, so that $\eta\le\mu$.
Then
\begin{align*}
 \sum_iw_i h(x_i)\ge(1-\eta)\kappa b-\eta\lambda b
 \ge\bigl((1-\mu)\kappa-\mu\lambda\bigr)b=0.
\end{align*}
Here the second inequality holds because the expression is decreasing in $\eta$.
If $b=0$ and $a>0$, the same argument uses the weight with $x_{i1}>0$ and replaces $b$ by $a$.

Now suppose $a,b>0$, and put $F=\|f\|_p$ and $L=a+b$.
Separate the agents into four groups according to whether each coordinate is strictly positive or nonpositive.
Let $G=\kappa L$ when $\mu\le\frac12$, and $G=\kappa F$ when $\mu>\frac12$.
When $\mu\le\frac12$, we have $\frac{\kappa}{\lambda}=\frac{\mu}{1-\mu}\le1$, so the stronger bound of Lemma~\ref{lem:planar-quadrants} is available.
That lemma bounds the contributions of the four groups by $G$, $\kappa b-\lambda a$, $\kappa a-\lambda b$, and $-\lambda F$, respectively.
Write their weights as $m_{--},m_{+-},m_{-+},m_{++}$, where a minus sign includes zero.
They sum to one and satisfy $m_{+-}+m_{++}\le\mu$ and $m_{-+}+m_{++}\le\mu$.

Consider the minimum weighted sum of these four numbers subject to the two positive-coordinate weights being at most $\mu$.
Changing either sign from nonpositive to positive never increases the corresponding number.
Indeed, $G\ge\kappa a,\kappa b$ because both $L$ and $F$ are at least $a,b$, and
\begin{align*}
 (\kappa b-\lambda a)-(-\lambda F)=\kappa b+\lambda(F-a)\ge0,
 \qquad
 (\kappa a-\lambda b)-(-\lambda F)=\kappa a+\lambda(F-b)\ge0.
\end{align*}
If the first positive-coordinate weight is $\eta<\mu$, transfer mass $\mu-\eta$ from the two groups with a nonpositive first coordinate to the corresponding groups with a positive first coordinate, leaving the second sign unchanged.
There is enough mass to do this because these source groups have total weight $1-\eta\ge\mu-\eta$.
The transfer does not increase the weighted sum and leaves the second positive-coordinate weight unchanged.
Repeat for the second coordinate, now leaving the first sign unchanged.
Both positive-coordinate weights then equal $\mu$.
These are transfers among four abstract weights to minimize a lower bound; no agent locations or mechanism outputs are changed.

Let $t$ be the weight in the group with two positive coordinates.
The two marginal equalities give $m_{+-}=m_{-+}=\mu-t$, and the total-weight equality gives $m_{--}=1-2\mu+t$.
Nonnegativity of these four weights is equivalent to $\max\{0,2\mu-1\}\le t\le\mu$.
The weighted sum of the bounds becomes
\begin{align}
 T(t)=(1-2\mu+t)G+(\mu-t)(\kappa-\lambda)L-t\lambda F.
 \label{eq:planar-sign-sum}
\end{align}
This expression is affine in $t$, so it is enough to check the two endpoints.

If $\mu\le\frac12$, substitute $G=\kappa L$ into \eqref{eq:planar-sign-sum}.
The endpoints are $0$ and $\mu$, and direct substitution gives
\begin{align*}
 T(0)&=\bigl((1-2\mu)\kappa+\mu(\kappa-\lambda)\bigr)L
       =\bigl((1-\mu)\kappa-\mu\lambda\bigr)L=0,\\
 T(\mu)&=(1-\mu)\kappa L-\mu\lambda F=\mu\lambda(L-F)\ge0.
\end{align*}
If $\mu>\frac12$, use $G=\kappa F$ instead.
The endpoints are now $2\mu-1$ and $\mu$.
Since $(1-\mu)(\kappa-\lambda)=(2\mu-1)\lambda$, we obtain
\begin{align*}
 T(2\mu-1)&=(1-\mu)(\kappa-\lambda)L-(2\mu-1)\lambda F
            =(2\mu-1)\lambda(L-F)\ge0,\\
 T(\mu)&=(1-\mu)\kappa F-\mu\lambda F=0.
\end{align*}
In both cases $L=a+b\ge(a^p+b^p)^{\frac{1}{p}}=F$, which explains the nonnegative signs.
Hence $T(t)$ is nonnegative throughout the feasible interval.
The transfers could only decrease the sum of the four group bounds, so that sum is also nonnegative for the original group weights.
Since $\sum_i w_i h(x_i)$ is at least this sum, the desired cost inequality follows.
\end{proof}

\begin{proof}[Proof of the upper bounds in Theorem~\ref{thm:planar}]
First take $1<p<+\infty$, and translate the mechanism output to the origin.
For consistency, use the prediction itself as the comparison point $f=\widehat f$.
In every coordinate with $f_j\ne0$, all prediction weight lies strictly on the side of the origin containing $f_j$.
The median condition allows augmented weight at most $\frac{1+c}{2}$ on that side.
After subtracting the prediction weight $c$, the real weight on that side is at most $\frac{1-c}{2}$.
Lemma~\ref{lem:planar-cost}, with $\mu=\frac{1-c}{2}$, therefore gives 
\begin{align*}
 \SC_p(\mathbf{x},\CMP_c(\mathbf{x},f))\le A_p(c)\SC_p(\mathbf{x},f).
\end{align*}
Indeed, the parameter in that lemma becomes $\frac{\mu}{1-\mu}=\frac{1-c}{1+c}$, so its factor is exactly $A_p(c)$.
When $f$ is optimal, $\SC_p(\mathbf{x},f)=\OPT_p(\mathbf{x})$, proving consistency.

For robustness, take any comparison point $f$ and any prediction.
In each coordinate, the real weight strictly on either side of the output is at most the augmented weight on that side, and hence at most $\frac{1+c}{2}$.
Apply Lemma~\ref{lem:planar-cost} with $\mu=\frac{1+c}{2}$ and then choose $f$ to be optimal.
Here $\frac{\mu}{1-\mu}=\frac{1+c}{1-c}$, and hence the lemma gives
\begin{align*}
 \SC_p(\mathbf{x},\CMP_c(\mathbf{x},\widehat f))
 \le\left(1+\left(\frac{1+c}{1-c}\right)^q\right)^{\frac{1}{q}}\SC_p(\mathbf{x},f)
 =B_p(c)\OPT_p(\mathbf{x}).
\end{align*}
These arguments use only the median inequalities, which hold for the selected lower median even when half of the augmented weight lies on one side.
The reflections in Lemma~\ref{lem:planar-cost} are used to compare distances and side weights, not to select a new mechanism output.

Finally, the cost comparisons above hold for every comparison point $f$: consistency uses prediction $f$, while robustness permits any prediction.
Fix the profile, prediction, and comparison point, and let $p$ tend to $1$ or to $+\infty$.
The mechanism output is independent of $p$, and every distance converges to the corresponding endpoint norm.
For fixed $r>0$, $\bigl(1+r^{q}\bigr)^{\frac{1}{q}}$ converges to $\max\{1,r\}$ as $p\downarrow1$ and to $1+r$ as $p\to+\infty$.
For the first limit, if $M=\max\{1,r\}$ then $M\le(1+r^q)^{\frac{1}{q}}\le2^{\frac{1}{q}}M$, and $q\to+\infty$ as $p\downarrow1$.
For the second, $q\to1$ as $p\to+\infty$, so ordinary continuity gives $1+r$.
Thus we may pass to the limit in the cost inequalities without assuming that an optimal facility varies continuously with $p$.
This proves the endpoint cost comparisons as well.
Choosing an optimum for the endpoint norm gives the stated values of $A_1$, $B_1$, $A_\infty$, and $B_\infty$ as upper bounds.
\end{proof}

\subsection{Tightness}

We next construct matching instances.
The following calculation also verifies that the prediction used for consistency is optimal.

\begin{lemma}\label{lem:planar-lower}
Fix $1<p<+\infty$ and $\mu\in(0,1)$.
Let $r=\frac{\mu}{1-\mu}$ and $x=r^{-\frac{1}{p-1}}$.
Place mass $\frac{1-\mu}{2}$ at each of $(-x,0)$ and $(x,0)$ and mass $\mu$ at $f=(0,1)$.
Then $f$ is optimal and
\begin{align}
 \frac{\SC_p(\mathbf{x},0)}{\OPT_p(\mathbf{x})}
 =\big(1+r^{q}\big)^{\frac{1}{q}}.
 \label{eq:planar-lower-ratio}
\end{align}
\end{lemma}

\begin{proof}
Write the cost of the two horizontal groups explicitly as
\begin{align*}
 C_H(y)=\frac{1-\mu}{2}\bigl(\|y-(-x,0)\|_p+\|y-(x,0)\|_p\bigr).
\end{align*}
Since $x>0$, both distances are differentiable at $f=(0,1)$.
Their first-coordinate derivatives have opposite signs and equal magnitudes, while both second-coordinate derivatives are $(1+x^p)^{-\frac{1}{q}}$.
Thus $\nabla C_H(f)=(0,g)$, where
\begin{align*}
 g=\frac{1-\mu}{(1+x^p)^{(p-1)/p}}
 =\mu\left(1+r^{q}\right)^{-\frac{1}{q}}\le\mu.
\end{align*}
For the displayed simplification, use $x^p=r^{-q}$ and $(1-\mu)r=\mu$, so that $(1+x^p)^{-\frac{1}{q}}=r(1+r^q)^{-\frac{1}{q}}$.
The supporting-hyperplane inequality for the convex differentiable function $C_H$ gives 
\begin{align*}
 C_H(y)\ge C_H(f)+\langle\nabla C_H(f),y-f\rangle=C_H(f)+g(y_2-1)
\end{align*} 
for every $y$.
The group at $f$ contributes 
\begin{align*}
 \mu\|y-f\|_p\ge\mu|y_2-1|\ge-g(y_2-1).
\end{align*}
Adding the two inequalities proves $\SC_p(\mathbf{x},y)\ge\SC_p(\mathbf{x},f)$, so $f$ is optimal.
Finally, the costs at the origin and at $f$ are $(1-\mu)x+\mu$ and $(1-\mu)\bigl(1+x^p\bigr)^{\frac{1}{p}}$.
Since $x^p=r^{-q}$ and $\frac{r}{x}=r^q$, dividing the numerator and denominator by $x$ gives
\begin{align*}
 \frac{\SC_p(\mathbf{x},0)}{\OPT_p(\mathbf{x})}
 =\frac{x+r}{(1+x^p)^{1/p}}
 =\frac{1+r^q}{(1+r^q)^{1/p}}
 =\big(1+r^q\big)^{\frac{1}{q}},
\end{align*}
which proves \eqref{eq:planar-lower-ratio}.
\end{proof}

\begin{proof}[Proof of tightness in Theorem~\ref{thm:planar}]
For $1<p<+\infty$, set $\mu=\frac{1-c}{2}$ in Lemma~\ref{lem:planar-lower} and use the correct prediction $f$.
In the vertical coordinate, the real mass at zero is $1-\mu=\frac{1+c}{2}$, exactly half of the augmented total weight.
Hence the lower median is zero.
In the horizontal coordinate, each nonzero side has weight $\frac{1-\mu}{2}<\frac{1+c}{2}$, so the median is also zero.
Equation~\eqref{eq:planar-lower-ratio} gives the ratio $A_p(c)$.

For robustness, set $\mu=\frac{1+c}{2}$ and use prediction $(0,0)$.
The augmented vertical mass at zero is $1-\mu+c=\frac{1+c}{2}$, so the lower median is again zero.
The horizontal median is zero for the same reason as above.
Equation~\eqref{eq:planar-lower-ratio} now gives $B_p(c)$.

For $p=\infty$, use the same two choices of $\mu$ and prediction, but put the horizontal groups at $(-1,0)$ and $(1,0)$.
The triangle inequality gives $\|y-(-1,0)\|_\infty+\|y-(1,0)\|_\infty\ge2$ for every $y$.
Thus the horizontal groups alone have cost at least $1-\mu$, and $f=(0,1)$ attains this cost while serving the remaining group at zero cost.
The mechanism still returns the origin, where the cost is one.
The ratio is $\frac{1}{1-\mu}$, yielding $A_\infty(c)$ and $B_\infty(c)$.

For $p=1$, consistency is already optimal.
For robustness, put mass $1-\mu$ at the origin and mass $\mu=\frac{1+c}{2}$ at $f=(0,1)$, with prediction at the origin.
Since $\mu\ge\frac12$, the triangle inequality gives, for every $y$,
\begin{align*}
 (1-\mu)\|y\|_1+\mu\|y-f\|_1
 \ge(1-\mu)\bigl(\|f\|_1-\|y-f\|_1\bigr)+\mu\|y-f\|_1
 \ge1-\mu.
\end{align*}
Equality holds at $y=f$, so $f$ is optimal.
The augmented mass at zero is exactly half of the total, so the lower median is zero.
The ratio is $\frac{\mu}{1-\mu}=\frac{1+c}{1-c}$.

To obtain ordinary finite profiles for any real $c$, choose $\mu_k=\frac{\lfloor2k\mu\rfloor}{2k}$ and take $k$ large enough that $\mu_k>0$.
Then $\mu_k\le\mu$ and $\mu_k\to\mu$.
In the three-location constructions, replace $\mu$ by $\mu_k$ and, for finite $p>1$, choose the corresponding $x_k$ from Lemma~\ref{lem:planar-lower}.
All real weights are rational: a population of $4k$ agents realizes mass $\mu_k$ at $f$ and mass $\frac{1-\mu_k}{2}$ at each horizontal point.
In the consistency construction the vertical mass at zero is $1-\mu_k\ge1-\mu=\frac{1+c}{2}$; in the robustness construction it is $1-\mu_k+c\ge1-\mu+c=\frac{1+c}{2}$.
There is no negative vertical mass, so the selected lower median is zero.
In the horizontal coordinate each nonzero side has mass $\frac{1-\mu_k}{2}<\frac{1+c}{2}$, and the positive mass at zero makes the cumulative weight reach half there.
Thus the horizontal median is also zero.
The point $f$ remains optimal by the same construction, and the ratios converge to the claimed bounds.
For the two-location $\ell_1$ construction, $\mu_k\ge\frac12$ as well, so the same argument applies.
Thus tightness holds for the suprema over ordinary finite profiles, without an integrality assumption on $cn$.
\end{proof}

We can now state the three-location reduction as a consequence of the upper bounds and matching instances.
This also avoids any assumption that a sequence of instance transformations terminates.

\begin{corollary}[Planar reduction]\label{cor:planar-reduction}
Fix $1<p<+\infty$ and $c\in[0,1)$.
For consistency, consider any instance with a correct prediction and positive optimal cost.
There is a weighted instance with an approximation ratio at least as large in which the real agents occupy $(-x,0)$, $(x,0)$, and $(0,1)$ for some $x>0$.
The mechanism returns $(0,0)$, and the optimal facility and prediction are $(0,1)$.
For robustness, the same statement holds for arbitrary initial predictions, with the new prediction at $(0,0)$.
If ordinary finite profiles are required, the new ratios can approach the respective worst-case bounds arbitrarily closely.
\end{corollary}

\begin{proof}
Theorem~\ref{thm:planar} bounds every consistency ratio by $A_p(c)$ and every robustness ratio by $B_p(c)$.
The corresponding weighted instances in Lemma~\ref{lem:planar-lower}, with the predictions specified in the tightness proof, attain these bounds and have the stated form.
In particular, if the starting ratio is $R$, the appropriate construction has ratio $A_p(c)\ge R$ for consistency or $B_p(c)\ge R$ for robustness.
This comparison uses the already proved bounds, not a sequence of transformations of the starting profile.
The rational approximations in that proof establish the assertion for ordinary finite profiles; they approach the weighted ratio, without requiring exact attainment by a finite profile.
\end{proof}

For $p=2$, the two formulas become $A_2(c)=\frac{\sqrt{2c^2+2}}{1+c}$ and $B_2(c)=\frac{\sqrt{2c^2+2}}{1-c}$, recovering the Euclidean guarantees of Agrawal et al.\ \cite{agrawal2022learning}.
The planar upper bound uses both coordinate mass constraints through the four sign groups in Lemma~\ref{lem:planar-cost}.
The next section treats arbitrary dimensions by a different aggregation of the coordinate constraints.

\section{Asymptotically Tight Bounds in High Dimensions}\label{sec:high-dimensional}

We now consider $\ell_p(\R^d)$ for arbitrary $d$.
We first prove dimension-independent upper bounds and then construct instances that approach these bounds as the dimension grows.

\subsection{Upper Bounds}

\begin{theorem}\label{thm:high-upper}
For every $1<p<+\infty$, $c\in[0,1)$, and $d\ge1$, the coordinate-wise median with predictions satisfies
\begin{align}
 \alpha_{p,d}(c)\le\Gamma_p\left(\frac{1-c}{2}\right),
 \qquad
 \beta_{p,d}(c)\le\Gamma_p\left(\frac{1+c}{2}\right).
 \label{eq:high-upper}
\end{align}
Here, for $\mu\in(0,1)$ and $q=\frac{p}{p-1}$, define
\begin{align}
 R_{p,\mu}(a)&=\frac{a^{1/p}+(\mu-a)/(1-\mu)}{(1-a)^{1/p}},
 &\Delta_p(\mu)&=\max_{0\le a\le\mu}R_{p,\mu}(a),
 &\Gamma_p(\mu)&=\left(1+\Delta_p(\mu)^q\right)^{\frac{1}{q}}.
 \label{eq:high-gamma}
\end{align}
\end{theorem}

We prove the two inequalities in \eqref{eq:high-upper} below.
The proof adapts the orthant relaxation of Gravin and Jia \cite{gravin2025approximation} to the biased coordinate medians induced by the prediction.
We give the full argument because the prediction changes the scalar constraint that determines the bound.

Translate the mechanism output to the origin and reflect coordinates so that an optimal facility $f$ is coordinate-wise nonnegative.
Reflection preserves the two median inequalities: at most half of the augmented weight lies strictly on either side in each coordinate.
We use these inequalities for the reflected data; we do not claim that reflection preserves the selected lower median when there is a tie.
For $f\ne0$, rescale all locations so that $\|f\|_p=1$; both costs scale by the same factor, leaving the approximation ratio unchanged.
The case $f=0$ is immediate.
Continue to normalize real weights to sum to one and denote them by $w_i$.
We may split an agent into several weighted copies at the same location and assign signs to their zero coordinates.
This changes neither social cost nor the coordinate-median inequalities.
For each copy a positive sign means a nonnegative coordinate, and a negative sign means a nonpositive coordinate; only zero coordinates admit either choice.
Write $[d]=\{1,\ldots,d\}$ for the coordinate indices.
For an agent $i$, let $S_i\subseteq[d]$ be the set of coordinates whose assigned sign is positive and define $z_i=\sum_{j\in S_i}f_j^p$.

\begin{lemma}\label{lem:mean-z}
In the consistency analysis, the zero coordinates can be split so that the weighted average $\sum_iw_i z_i$ is $\frac{1-c}{2}$.
In the robustness analysis, they can be split so that this average is at most $\frac{1+c}{2}$.
\end{lemma}

\begin{proof}
Consider a coordinate $j$ with $f_j>0$.
Let $P_j,N_j,Z_j$ be the real weights with strictly positive, strictly negative, and zero coordinates, respectively, so that $P_j+N_j+Z_j=1$.
Under a correct prediction, all prediction weight is strictly positive in this coordinate.
The median inequalities give
\begin{align*}
 P_j+c\le\frac{1+c}{2},\qquad N_j\le\frac{1+c}{2},
 \qquad\text{hence}\qquad
 P_j\le\frac{1-c}{2}\le1-N_j=P_j+Z_j.
\end{align*}
Assign positive sign to zero-coordinate mass $\frac{1-c}{2}-P_j$ and negative sign to the rest.
The displayed bounds show that the required mass lies between $0$ and $Z_j$.
These choices can be made coordinate by coordinate: whenever a copy must be split, give both new copies all previously assigned signs, so earlier signed marginal weights remain unchanged.
There are finitely many coordinates, so this procedure produces finitely many weighted copies.
Under an arbitrary prediction, assign every zero coordinate a negative sign.
The positive signed real mass then equals $P_j$, which is at most the augmented positive mass and hence at most $\frac{1+c}{2}$.
Coordinates with $f_j=0$ make no contribution to $z_i$.
Let $P_j^+$ denote the positive signed real mass after splitting.
Interchanging the two finite sums gives 
\begin{align*}
 \sum_iw_i z_i=\sum_iw_i\sum_{j\in S_i}f_j^p=\sum_jf_j^p P_j^+.
\end{align*}
Since $\sum_jf_j^p=\|f\|_p^p=1$, we obtain
\begin{align}
 \sum_iw_i z_i=\frac{1-c}{2}\quad\text{for consistency},
 \qquad
 \sum_iw_i z_i\le\frac{1+c}{2}\quad\text{for robustness}.
 \label{eq:mean-z}
\end{align}
\end{proof}

Weighted copies are used only in the analysis; they do not change the original instance or the mechanism output.

\begin{lemma}[Orthant inequality]\label{lem:orthant}
For every $\lambda\in(0,1)$, let $K_\lambda=\bigl(1-\lambda^{q}\bigr)^{\frac{1}{q}}$.
Then every agent satisfies
\begin{align}
 \|x_i-f\|_p-\lambda\|x_i\|_p
 \ge K_\lambda(1-z_i)^{\frac{1}{p}}-\lambda z_i^{\frac{1}{p}}.
 \label{eq:orthant}
\end{align}
\end{lemma}

\begin{proof}
For $S\subseteq[d]$, write $S^c=[d]\setminus S$ and let $x_S$ denote the subvector of $x$ consisting of the coordinates in $S$.
Fix $i$ and write $S=S_i$, $a=\|f_S\|_p=z_i^{\frac{1}{p}}$, $b=\|f_{S^c}\|_p=(1-z_i)^{\frac{1}{p}}$, $u=\|x_{i,S}\|_p$, and $v=\|x_{i,S^c}\|_p$.
On $S$, the reverse triangle inequality gives $\|x_{i,S}-f_S\|_p\ge|u-a|$.
On $S^c$, we have $x_{ij}\le0$ and $f_j\ge0$, including any zeros assigned a negative sign.
Thus $|x_{ij}-f_j|=|x_{ij}|+f_j$, and
\begin{align*}
 \|x_{i,S^c}-f_{S^c}\|_p^p
 =\sum_{j\in S^c}(|x_{ij}|+f_j)^p
 \ge\sum_{j\in S^c}(|x_{ij}|^p+f_j^p)=v^p+b^p.
\end{align*}
Let $T=\bigl(|u-a|^p+v^p\bigr)^{\frac{1}{p}}$.
Adding the estimates on $S$ and $S^c$ gives 
\begin{align*}
 \|x_i-f\|_p^p\ge|u-a|^p+v^p+b^p=T^p+b^p.
\end{align*}
For the cost at the origin, apply the triangle inequality to the two-dimensional decomposition $(u,v)=(a,0)+(u-a,v)$:
\begin{align*}
 \|x_i\|_p=(u^p+v^p)^{\frac{1}{p}}
 \le a+\bigl(|u-a|^p+v^p\bigr)^{\frac{1}{p}}=a+T.
\end{align*}
Finally, H\"older's inequality applied to $(T,b)$ and $(\lambda,K_\lambda)$ gives
\begin{align*}
 \lambda T+K_\lambda b
 \le(T^p+b^p)^{\frac{1}{p}}(\lambda^q+K_\lambda^q)^{\frac{1}{q}}
 =(T^p+b^p)^{\frac{1}{p}}.
\end{align*}
Combining the three estimates and substituting the definitions of $a$ and $b$, we obtain
\begin{align*}
 \|x_i-f\|_p-\lambda\|x_i\|_p
 \ge(T^p+b^p)^{\frac{1}{p}}-\lambda(a+T)
 \ge K_\lambda b-\lambda a
 =K_\lambda(1-z_i)^{\frac{1}{p}}-\lambda z_i^{\frac{1}{p}}.
\end{align*}
Empty $S$ or $S^c$ causes no difficulty: its subvector has norm zero, and the same inequalities apply.
\end{proof}

Set $\delta=\frac{K_\lambda}{\lambda}$ and define $h_\delta(z)=\delta(1-z)^{\frac{1}{p}}-z^{\frac{1}{p}}$.
Summing \eqref{eq:orthant} shows that it is enough to make the average of $h_\delta(z_i)$ nonnegative.
The next lemma gives the exact condition.

\begin{lemma}\label{lem:convex-envelope}
Fix $1<p<+\infty$, $\delta>0$, and $\mu\in(0,1)$.
Let $z_i\in[0,1]$ have weights $w_i\ge0$ with $\sum_iw_i=1$ and $\sum_iw_i z_i\le\mu$.
Then $\sum_iw_i h_\delta(z_i)\ge0$ whenever $\delta\ge\Delta_p(\mu)$.
Conversely, if $\delta<\Delta_p(\mu)$, there is a distribution supported on at most two points with mean $\mu$ for which the average of $h_\delta$ is negative.
\end{lemma}

\begin{proof}
Write $s=\frac{1}{p}\in(0,1)$.
Direct differentiation gives
\begin{align*}
 h_\delta'(z)=-s\delta(1-z)^{s-1}-sz^{s-1}<0,
 \qquad
 h_\delta''(z)=s(1-s)\bigl(z^{s-2}-\delta(1-z)^{s-2}\bigr).
\end{align*}
Here $\delta>0$, as in its definition above.
The equation $h_\delta''(z)=0$ is equivalent to $\bigl(\frac{z}{1-z}\bigr)^{s-2}=\delta$.
Its left-hand side decreases strictly from $+\infty$ to $0$, so there is a unique solution $b_0\in(0,1)$.
Consequently $h_\delta$ is strictly convex on $(0,b_0)$ and strictly concave on $(b_0,1)$.
To identify the lower convex envelope, define
\begin{align*}
 Q(a)=h_\delta(a)+(1-a)h_\delta'(a)-h_\delta(1).
\end{align*}
Differentiating and cancelling the two first-derivative terms gives $Q'(a)=(1-a)h_\delta''(a)$.
It follows that $Q$ increases up to $b_0$ and decreases afterward.
To check the endpoint signs, expand its definition:
\begin{align*}
 Q(a)=1-a^s-s(1-a)a^{s-1}+\delta(1-s)(1-a)^s.
\end{align*}
As $a\downarrow0$, the term $-s(1-a)a^{s-1}$ tends to $-\infty$, while all other terms remain bounded.
As $a\uparrow1$, the expanded expression gives $Q(a)\to0$: $1-a^s\to0$, $(1-a)a^{s-1}\to0$, and $(1-a)^s\to0$.
Since $Q$ is strictly decreasing on $(b_0,1)$, it is positive throughout this interval, and in particular at $b_0$.
Thus it has a unique zero $a_0\in(0,b_0)$.

Let $\ell$ be the line through $(a_0,h_\delta(a_0))$ and $(1,h_\delta(1))$.
The equation $Q(a_0)=0$ says exactly that its slope is $h_\delta'(a_0)$.
For $D(z)=h_\delta(z)-\ell(z)$, we therefore have $D(a_0)=D(1)=0$ and $D'(a_0)=0$.
On $(a_0,b_0)$, $D'$ is increasing and positive.
On $(b_0,1)$, it decreases strictly to $-\infty$, so it crosses zero exactly once.
Thus $D$ first increases and then decreases to zero, proving $D(z)\ge0$ on $[a_0,1]$.

Define $g=h_\delta$ on $[0,a_0]$ and $g=\ell$ on $[a_0,1]$.
The first part is convex, and its derivative at $a_0$ equals the slope of the second part, so $g$ is convex on the whole interval.
Also $g\le h_\delta$.
For any other convex function $\widetilde g\le h_\delta$ and any $z\in[a_0,1]$, convexity gives
\begin{align*}
 \widetilde g(z)\le\frac{1-z}{1-a_0}\widetilde g(a_0)
 +\frac{z-a_0}{1-a_0}\widetilde g(1)
 \le\frac{1-z}{1-a_0}h_\delta(a_0)+\frac{z-a_0}{1-a_0}h_\delta(1)=g(z).
\end{align*}
On $[0,a_0]$, the inequality $\widetilde g\le g$ follows directly from $g=h_\delta$.
Hence $g$ is the lower convex envelope, that is, the greatest convex function bounded above by $h_\delta$.
Equivalently, at a prescribed mean $\mu$ its value is
\begin{align*}
 \min_{0\le a\le\mu}
 \left\{\frac{1-\mu}{1-a}h_\delta(a)
 +\frac{\mu-a}{1-a}h_\delta(1)\right\}.
\end{align*}
To see this formula, a distribution supported at $a$ and $1$ with mean $\mu$ must have weights $\frac{1-\mu}{1-a}$ and $\frac{\mu-a}{1-a}$.
Its expected value is at least $g(\mu)$ by convexity.
Equality is attained by $a=\mu$ when $\mu\le a_0$, and by $a=a_0$ when $\mu>a_0$.
Since $h_\delta(1)=-1$, the expression indexed by $a$ equals
\begin{align*}
 \frac{(1-\mu)\bigl(\delta(1-a)^{1/p}-a^{1/p}\bigr)-(\mu-a)}{1-a}.
\end{align*}
Both $1-a$ and $1-\mu$ are positive, so this is nonnegative exactly when
\begin{align*}
 \delta\ge\frac{a^{1/p}+(\mu-a)/(1-\mu)}{(1-a)^{1/p}}=R_{p,\mu}(a).
\end{align*}
It is nonnegative for every $a\in[0,\mu]$ exactly when $\delta\ge\Delta_p(\mu)$.
The function $g$ is decreasing: it agrees with the decreasing function $h_\delta$ before $a_0$ and then has constant negative slope $h_\delta'(a_0)$.
For any weights $w_i\ge0$ summing to one and mean $m=\sum_iw_i z_i\le\mu$, Jensen's inequality gives
\begin{align*}
 \sum_iw_i h_\delta(z_i)\ge\sum_iw_i g(z_i)
 \ge g\left(\sum_iw_i z_i\right)=g(m)\ge g(\mu)\ge0.
\end{align*}
This proves the first statement.
For the converse, choose an $a\in[0,\mu]$ with $\delta<R_{p,\mu}(a)$ and put the two displayed weights at $a$ and $1$.
They are nonnegative, sum to one, and have mean $\frac{a(1-\mu)+\mu-a}{1-a}=\mu$.
The expected value is negative by the preceding calculation.
If $a=\mu$, the weight at $1$ is zero, so the distribution is supported on just one point.
\end{proof}

\begin{proof}[Proof of Theorem~\ref{thm:high-upper}]
Fix $\mu$ and take $\delta=\Delta_p(\mu)>0$ and $\lambda=\bigl(1+\delta^q\bigr)^{-\frac{1}{q}}\in(0,1)$.
These choices give $K_\lambda^q=1-\lambda^q=\frac{\delta^q}{1+\delta^q}=\delta^q\lambda^q$, so $K_\lambda=\delta\lambda$.
By Lemmas~\ref{lem:orthant} and \ref{lem:convex-envelope},
\begin{align*}
 \SC_p(\mathbf{x},f)-\lambda\SC_p(\mathbf{x},0)
 \ge \lambda\sum_i w_i h_\delta(z_i)\ge0.
\end{align*}
Here Lemma~\ref{lem:mean-z} supplies the mean constraint, and Lemma~\ref{lem:convex-envelope} makes the last sum nonnegative.
Since $f$ is optimal, division by $\lambda\OPT_p(\mathbf{x})>0$ gives
\begin{align*}
 \frac{\SC_p(\mathbf{x},0)}{\OPT_p(\mathbf{x})}
 \le\frac{1}{\lambda}=\big(1+\Delta_p(\mu)^q\big)^{\frac{1}{q}}=\Gamma_p(\mu).
\end{align*}
Use $\mu=\frac{1-c}{2}$ for consistency and $\mu=\frac{1+c}{2}$ for robustness, as supplied by the two relations in \eqref{eq:mean-z}.
\end{proof}

The maximization in \eqref{eq:high-gamma} is one-dimensional and has a simple characterization.

\begin{proposition}\label{prop:maximizer}
Fix $1<p<+\infty$ and $\mu\in(0,1)$, and let $\overline a_p(\mu)$ be the unique solution in $(0,1)$ of
\begin{align}
 (1-\mu)a^{-\frac{p-1}{p}}+(p-1)a+\mu-p=0.
 \label{eq:stationary}
\end{align}
Then a maximizer of $R_{p,\mu}$ is $a_p^*(\mu)=\min\{\mu,\overline a_p(\mu)\}$.
In particular, the maximizer is the endpoint $\mu$ exactly when $\mu\le p^{-\frac{p}{p-1}}$.
\end{proposition}

\begin{proof}
Let $F(a)$ denote the left-hand side of \eqref{eq:stationary}.
For $0<a<1$, differentiating the numerator and denominator gives
\begin{align*}
 R_{p,\mu}'(a)=\frac{(1-\mu)a^{-(p-1)/p}+(p-1)a+\mu-p}{p(1-\mu)(1-a)^{1+1/p}}
 =\frac{F(a)}{p(1-\mu)(1-a)^{1+1/p}}.
\end{align*}
Thus $R_{p,\mu}'$ and $F$ have the same sign, since the denominator is positive.
Moreover, $F(a)\to+\infty$ as $a\downarrow0$, $F(1)=0$, and
\begin{align*}
 F'(a)=(p-1)\left(1-\frac{1-\mu}{p}a^{-\frac{2p-1}{p}}\right)
\end{align*}
is strictly increasing from $-\infty$ to a positive value at one.
In particular, $F'(1)=(p-1)\bigl(1-\frac{1-\mu}{p}\bigr)>0$, so $F(a)<F(1)=0$ just to the left of one.
The strict increase of $F'$ implies a single minimum of $F$, and the limit $F(a)\to+\infty$ at zero then gives the unique earlier crossing.
Denote this root by $\overline a_p(\mu)$; $F$ is positive before it and negative between it and $1$.
Consequently $R_{p,\mu}$ increases on $(0,\overline a_p(\mu))$ and decreases on $(\overline a_p(\mu),1)$.
Continuity at zero shows that its maximum on $[0,\mu]$ occurs at $\min\{\mu,\overline a_p(\mu)\}$, which is strictly positive.
Finally, $F(\mu)=(1-\mu)\bigl(\mu^{-\frac{p-1}{p}}-p\bigr)$, which is nonnegative exactly when $\mu\le p^{-\frac{p}{p-1}}$.
This is precisely the condition that the interval $[0,\mu]$ ends before the interior maximizer.
\end{proof}

For $p=2$, factoring the cubic obtained from \eqref{eq:stationary} gives $\overline a_2(\mu)=\bigl(\frac{\sqrt{5-4\mu}-1}{2}\bigr)^2$.
Substitution in \eqref{eq:high-gamma} recovers the piecewise Euclidean bounds of Gravin and Jia \cite{gravin2025approximation}.
We next establish the upper bounds for the endpoint norms.

\begin{theorem}\label{thm:endpoints}
For every dimension $d$ and $c\in[0,1)$, $\alpha_{1,d}(c)=1$ and $\beta_{1,d}(c)=\frac{1+c}{1-c}$.
Moreover,
\begin{align}
 \alpha_{\infty,d}(c)\le\frac{3+c}{1+c},
 \qquad
 \beta_{\infty,d}(c)\le\frac{3-c}{1-c},
 \label{eq:infty-upper}
\end{align}
and both inequalities are asymptotically tight as $d\to\infty$.
\end{theorem}

\begin{proof}[Proof of the upper bounds in Theorem~\ref{thm:endpoints}]
The $\ell_1$ objective separates as $\SC_1(\mathbf{x},y)=\sum_{j=1}^d\sum_iw_i|x_{ij}-y_j|$.
A point is optimal exactly when each of its coordinates minimizes the corresponding one-dimensional sum.
For each coordinate, the endpoint cost comparison in Theorem~\ref{thm:planar} applies by embedding that one-dimensional instance on an axis in the plane.
For a correct prediction it gives $\sum_iw_i|x_{ij}-m_j|\le\sum_iw_i|x_{ij}-f_j|$, where $m$ is the mechanism output.
For an arbitrary prediction the same comparison has factor $\frac{1+c}{1-c}$.
Summing over coordinates gives the claimed upper bounds.

For $p=\infty$, translate the mechanism output to zero and reflect coordinates so that an optimum $f$ is nonnegative.
Suppose that, for every coordinate $j$ with $f_j>0$, the agents with $x_{ij}\le0$ have total real weight at least $r\in(0,1]$.
If $f=0$, the mechanism is already optimal, so assume $\|f\|_\infty>0$.
Choose a coordinate $j$ with $f_j=\|f\|_\infty$.
Every agent with $x_{ij}\le0$ has $\|x_i-f\|_\infty\ge|x_{ij}-f_j|\ge f_j$.
Since these agents have total real weight at least $r$, we obtain
\begin{align*}
 \SC_\infty(\mathbf{x},f)\ge\sum_{i:x_{ij}\le0}w_i|x_{ij}-f_j|
 \ge r f_j=r\|f\|_\infty.
\end{align*}
The triangle inequality gives $\|x_i\|_\infty\le\|x_i-f\|_\infty+\|f\|_\infty$ for every agent.
Summing with weights and using $\sum_iw_i=1$ therefore yields
\begin{align*}
 \SC_\infty(\mathbf{x},0)\le\SC_\infty(\mathbf{x},f)+\|f\|_\infty
 \le\left(1+\frac{1}{r}\right)\SC_\infty(\mathbf{x},f).
\end{align*}
Under a correct prediction, the augmented median condition gives $r=\frac{1+c}{2}$; under an arbitrary prediction it gives $r=\frac{1-c}{2}$.
This proves \eqref{eq:infty-upper}.
\end{proof}

\subsection{Asymptotic Tightness}

We next prove that the bounds in Theorem~\ref{thm:high-upper} cannot be improved uniformly over the dimension.

\begin{theorem}\label{thm:high-tight}
For every fixed $1<p<+\infty$ and $c\in[0,1)$,
\begin{align}
 \lim_{d\to\infty}\alpha_{p,d}(c)=\Gamma_p\left(\frac{1-c}{2}\right),
 \qquad
 \lim_{d\to\infty}\beta_{p,d}(c)=\Gamma_p\left(\frac{1+c}{2}\right).
 \label{eq:high-tight}
\end{align}
\end{theorem}

\begin{proof}
Fix $\mu\in(0,1)$ and choose a maximizer $a$ of $R_{p,\mu}$.
Proposition~\ref{prop:maximizer} shows that $0<a\le\mu<1$, so the denominators below are nonzero.
We first assume that $ad$ is an integer.
In fact, the construction and the optimality calculation work for any $0<a\le\mu$ if we set $\delta=R_{p,\mu}(a)$; this observation will let us approximate $a$ and $\mu$ by rational numbers at the end.
Let $\delta=\Delta_p(\mu)$, $\lambda=\bigl(1+\delta^{q}\bigr)^{-\frac{1}{q}}$, and
\begin{align}
 \theta=\frac{1-\mu}{1-a},\qquad
 \phi=\frac{\mu-a}{1-a},\qquad
 t=1+\left(\frac{1-a}{a}\right)^{\frac{1}{p}}\delta^{-\frac{1}{p-1}}.
 \label{eq:lower-parameters}
\end{align}
We have $\theta>0$, $\phi\ge0$, and $t>1$.
If $\phi=0$, no agents are placed at $f$.
Let $\one_d=(1,\ldots,1)\in\R^d$ and, for $S\subseteq[d]$, let $\one_S\in\R^d$ have coordinate one on $S$ and zero elsewhere.
Set $f=\one_d$.
Here we do not normalize $\|f\|_p$ to one; the common scale factor $d^{\frac{1}{p}}$ cancels from the cost ratio.
Place total mass $\theta$ uniformly over the points $t\one_S$, where $S$ ranges over all $ad$-element subsets of $[d]$, and place mass $\phi$ at $f$.
A fixed coordinate belongs to $\binom{d-1}{ad-1}$ of the $\binom{d}{ad}$ subsets, so the fraction of sparse points positive in that coordinate is $\frac{ad}{d}=a$.
The parameter identities are
\begin{align*}
 \theta+\phi=\frac{1-\mu+\mu-a}{1-a}=1,
 \qquad
 \theta a+\phi=\frac{a(1-\mu)+\mu-a}{1-a}=\mu.
\end{align*}
Thus every coordinate has positive real mass $\mu$ and zero real mass $1-\mu$.

For consistency, take $\mu=\frac{1-c}{2}$ and prediction $f$.
The augmented positive mass is $\mu+c=\frac{1+c}{2}$, exactly half of the augmented total mass.
There is no negative mass, and the zero mass $1-\mu=\frac{1+c}{2}$ reaches the median threshold, so the lower median is zero.
For robustness, take $\mu=\frac{1+c}{2}$ and prediction $-f$.
In every coordinate, the negative mass is $c<\frac{1+c}{2}$, and the mass at nonpositive coordinates is $c+1-\mu=\frac{1+c}{2}$.
Thus the threshold is first reached at zero.
This also covers $c=0$, when there is no negative prediction weight.
Hence $\CMP_c$ returns the origin in both constructions.

It remains to show that $f$ is an optimal facility.
Write $r=t-1$.
Let $D=(ar^p+1-a)^{\frac{1}{p}}$, so every sparse point is at distance $d^{\frac{1}{p}}D$ from $f$.
For the function $y\mapsto\|y-t\one_S\|_p$, its $j$-th derivative at $f$ is $-\frac{r^{p-1}}{d^{1/q}D^{p-1}}$ if $j\in S$, and $\frac{1}{d^{1/q}D^{p-1}}$ otherwise.
Averaging over the subsets gives the same derivative in every coordinate, namely $\frac{(1-a)-ar^{p-1}}{d^{1/q}D^{p-1}}$.

The following calculation shows that this averaged derivative is nonnegative.
The choice of $r=t-1$ gives
\begin{align*}
 r^{p-1}=\left(\frac{1-a}{a}\right)^{\frac{1}{q}}\frac{1}{\delta},
 \qquad
 D=(1-a)^{\frac{1}{p}}(1+\delta^{-q})^{\frac{1}{p}},
 \qquad
 (1+\delta^{-q})^{-\frac{1}{q}}=\lambda\delta.
\end{align*}
Using these identities, we obtain
\begin{align*}
 \frac{(1-a)-ar^{p-1}}{D^{p-1}}
 =\lambda\bigl(\delta(1-a)^{\frac{1}{p}}-a^{\frac{1}{p}}\bigr)
 =\lambda\frac{\mu-a}{1-\mu}
 =\lambda\frac{\phi}{\theta}\ge0.
\end{align*}
The second equality follows by multiplying $\delta=R_{p,\mu}(a)$ by $(1-a)^{\frac{1}{p}}$.
Thus, if $g$ denotes the gradient of the total sparse-point cost $C_H$ at $f$, every coordinate of $g$ equals $\frac{\lambda\phi}{d^{1/q}}$ and $\|g\|_q=\lambda\phi\le\phi$.

As in the planar construction, convexity gives $C_H(y)\ge C_H(f)+\langle g,y-f\rangle$.
H\"older's inequality gives $\langle g,y-f\rangle\ge-\|g\|_q\|y-f\|_p$, so
\begin{align*}
 \SC_p(\mathbf{x},y)=C_H(y)+\phi\|y-f\|_p
 \ge C_H(f)+(\phi-\|g\|_q)\|y-f\|_p
 \ge C_H(f)=\SC_p(\mathbf{x},f).
\end{align*}
This proves global optimality, including $\phi=0$, when $g=0$.

We next compute the ratio.
For a sparse point, the distance to the origin is $t(ad)^{\frac{1}{p}}$, and the distance to $f$ is $d^{\frac{1}{p}}D$.
Put $A=a^{\frac{1}{p}}$ and $B=(1-a)^{\frac{1}{p}}$, so that $rA=B\delta^{-\frac{1}{p-1}}$.
Using $(1+\delta^{-q})^{-\frac{1}{q}}=\lambda\delta$ and $1-q=-\frac{1}{p-1}$ gives
\begin{align*}
 (1+\delta^{-q})^{\frac{1}{p}}
 =(1+\delta^{-q})(1+\delta^{-q})^{-\frac{1}{q}}
 =\lambda\delta(1+\delta^{-q})
 =\lambda\delta+\lambda\delta^{-\frac{1}{p-1}}.
\end{align*}
Consequently,
\begin{align*}
 D-\lambda tA
 =B(1+\delta^{-q})^{\frac{1}{p}}-\lambda A-\lambda B\delta^{-\frac{1}{p-1}}
 =\lambda(\delta B-A)=\lambda\frac{\phi}{\theta}.
\end{align*}
Multiplying by $d^{\frac{1}{p}}=\|f\|_p$, each sparse point therefore satisfies
\begin{align*}
 \|t\one_S-f\|_p-\lambda\|t\one_S\|_p
 =\lambda\frac{\phi}{\theta}\|f\|_p,
\end{align*}
whereas the point $f$ satisfies $\|f-f\|_p-\lambda\|f\|_p=-\lambda\|f\|_p$.
The sparse group contributes $\theta\lambda\frac{\phi}{\theta}\|f\|_p$ to the cost difference, and the group at $f$ contributes $-\phi\lambda\|f\|_p$.
They cancel, yielding
\begin{align*}
 \SC_p(\mathbf{x},f)=\lambda\SC_p(\mathbf{x},0),
 \qquad
 \frac{\SC_p(\mathbf{x},0)}{\OPT_p(\mathbf{x})}=\Gamma_p(\mu).
\end{align*}

For every sufficiently large dimension $d$, set $a_d=\frac{\lfloor da\rfloor}{d}$ and $\mu_d=\frac{\lfloor d\mu\rfloor}{d}$.
Then $0<a_d\le\mu_d<1$, $a_d d$ is an integer, and $(a_d,\mu_d)\to(a,\mu)$.
Use these parameters in \eqref{eq:lower-parameters}, now with $\delta_d=R_{p,\mu_d}(a_d)$ and $\lambda_d=(1+\delta_d^q)^{-\frac{1}{q}}$.
The optimality and cost calculations above use only the identity $\delta_d=R_{p,\mu_d}(a_d)$, not that $a_d$ is a maximizer, so they hold exactly for each such $d$.

Keep the original confidence parameter $c$ and the same predictions.
Since $\mu_d\le\mu$, in the consistency construction the zero mass is $1-\mu_d\ge\frac{1+c}{2}$.
In the robustness construction the negative mass is still $c<\frac{1+c}{2}$ and the nonpositive mass is $c+1-\mu_d\ge\frac{1+c}{2}$.
Thus the selected median remains zero in both cases.
Moreover, $\theta_d=\frac{1-\mu_d}{1-a_d}$ and $\phi_d=\frac{\mu_d-a_d}{1-a_d}$ are rational.
Each sparse point has rational weight $\frac{\theta_d}{\binom{d}{a_d d}}$, so choosing a common denominator realizes the construction by ordinary agents exactly.

Finally, $R_{p,\mu}(a)$ is continuous for $0<a\le\mu<1$, and hence
\begin{align*}
 \frac{\SC_p(\mathbf{x}^{(d)},0)}{\OPT_p(\mathbf{x}^{(d)})}
 =(1+\delta_d^q)^{\frac{1}{q}}
 \longrightarrow(1+R_{p,\mu}(a)^q)^{\frac{1}{q}}=\Gamma_p(\mu).
\end{align*}
There is such a profile in every sufficiently large dimension, not just a subsequence.
Substituting the consistency and robustness values of $\mu$ and combining these lower bounds with Theorem~\ref{thm:high-upper} proves \eqref{eq:high-tight}.
\end{proof}

We finish by proving tightness for the endpoint norms.

\begin{proof}[Proof of tightness in Theorem~\ref{thm:endpoints}]
For $p=1$, the matching bounds already follow from the one-dimensional constructions.
The two-location $\ell_1$ construction in the planar tightness proof lies on one axis and can be embedded in any dimension, so its finite-profile approximations give the matching robustness lower bound.
Consistency cannot be less than one by the definition of the optimum.

For $p=\infty$, fix $\mu\in(0,1)$ and choose $d\ge2$ large enough that $d\mu\ge1$; let $f=\one_d$.
Let $e_j$ be the vector with coordinate one at $j$ and zero elsewhere.
Put total mass $\theta_d=\frac{1-\mu}{1-1/d}$ uniformly at $2e_1,\ldots,2e_d$ and the remaining mass $\phi_d=\frac{d\mu-1}{d-1}$ at $f$.
The positive mass in every coordinate is $\mu$.
For every $y$ and pair $j<k$, the triangle inequality gives $\|y-2e_j\|_\infty+\|y-2e_k\|_\infty\ge\|2e_j-2e_k\|_\infty=2$.
Each distance occurs in exactly $d-1$ pairs, so
\begin{align*}
 (d-1)\sum_{j=1}^d\|y-2e_j\|_\infty\ge2\binom d2=d(d-1),
 \qquad \frac{1}{d}\sum_{j=1}^d\|y-2e_j\|_\infty\ge1.
\end{align*}
At $f$, all these distances equal one and the group at $f$ pays zero.
Hence $\OPT_\infty(\mathbf{x})=\theta_d$, whereas $\SC_\infty(\mathbf{x},0)=2\theta_d+\phi_d=1+\theta_d$.
The ratio is therefore
\begin{align*}
 \frac{\SC_\infty(\mathbf{x},0)}{\OPT_\infty(\mathbf{x})}
 =1+\frac{1}{\theta_d}=1+\frac{1-1/d}{1-\mu}
 \longrightarrow1+\frac{1}{1-\mu}.
\end{align*}
Take $\mu=\frac{1-c}{2}$ with prediction $f$ or $\mu=\frac{1+c}{2}$ with prediction $-f$.
The positive real mass in each coordinate is $\theta_d/d+\phi_d=\mu$, so the same cumulative-weight calculations as in Theorem~\ref{thm:high-tight} show that the output is zero.
For ordinary finite profiles, replace $\mu$ by $\mu_d=\frac{\lfloor d\mu\rfloor}{d}$.
For sufficiently large $d$ we have $\mu_d\ge\frac{1}{d}$, so both masses are nonnegative and rational.
The positive mass only decreases, preserving the selected median at zero, and the preceding optimality proof still applies.
Since $\mu_d\to\mu$, the ratios have the same limits.
These limits are $\frac{3+c}{1+c}$ and $\frac{3-c}{1-c}$, respectively, and together with \eqref{eq:infty-upper} prove asymptotic tightness.
\end{proof}

At $c=0$, Theorems~\ref{thm:high-upper}, \ref{thm:endpoints}, and \ref{thm:high-tight} recover the prediction-free dimension-independent bounds, including $\Gamma_p\bigl(\frac{1}{2}\bigr)$ for $1<p<+\infty$ and the limiting factor $3$ for $p=\infty$ \cite{gravin2025approximation}.
The planar values can be strictly smaller because the scalar relaxation in Lemma~\ref{lem:convex-envelope} does not retain all the relations between coordinate-sign groups used in the planar proof.

\section{Conclusion}\label{sec:conclusion}

We studied the coordinate-wise median with predictions for utilitarian facility location in $\ell_p$ spaces.
In two dimensions, we established exact consistency and robustness bounds for every $p\in[1,+\infty]$, extending the Euclidean guarantees of Agrawal et al.\ \cite{agrawal2022learning}.
The proof combines the two coordinate-median constraints with distance inequalities for four sign groups, and matching instances establish tightness.
In high dimensions, we obtained dimension-independent bounds through a scalar mean constraint and constructed symmetric instances whose ratios approach these bounds as the dimension grows.
These results extend the prediction-free analysis of Gravin and Jia \cite{gravin2025approximation} and recover their learning-augmented Euclidean bounds.
We also determined the exact guarantees for $p=1$ in every dimension and asymptotically tight bounds for $p=+\infty$.

A natural direction is to determine the exact consistency and robustness for each fixed dimension $d\ge3$ and $1<p\le+\infty$.
Our high-dimensional analysis summarizes the coordinate constraints by a single mean constraint, losing information about which sign patterns can occur together.
Retaining more of this information may yield sharper dimension-dependent bounds and clarify the gap between the planar and asymptotic guarantees.

\paragraph{AI Use Statement.}
We used OpenAI GPT-based tools, including ChatGPT and Codex, to assist in developing and checking intermediate lemmas, proofs, and calculations.
These tools played a substantial role in extending the planar reduction and the high-dimensional upper- and lower-bound analyses.
The authors independently checked all AI-assisted claims, proofs, and calculations and take full responsibility for their correctness and for the paper's content.

\bibliographystyle{plain}
\bibliography{mybibfile}

\end{document}